\documentclass[11pt,leqno]{amsart}

\usepackage{amsmath}
\usepackage{amssymb}
\usepackage{a4wide}
\usepackage{bbm}
\usepackage{graphics}
\usepackage{epsfig}
\usepackage{hyperref}
\usepackage{amsmath, amssymb, latexsym, amscd, amsthm,amsfonts,amstext}
\usepackage[mathscr]{eucal}
\usepackage{appendix}
\usepackage[active]{srcltx}
\include{srctex}

\newtheorem{theorem}{Theorem}[section]
\newtheorem{lemma}[theorem]{Lemma}
\newtheorem{proposition}[theorem]{Proposition}
\newtheorem{definition}[theorem]{Definition}
\newtheorem{remark}[theorem]{Remark}
\newtheorem{corollary}[theorem]{Corollary}

\newcounter{hypo}

\makeatletter
 \@addtoreset{equation}{section}
 \makeatother
 
\title[the Schr\"odinger operators with constant magnetic fields.]{ Trace asymptotics formula for  the Schr\"odinger operators with constant magnetic fields.}

\begin{document}

\author{Mouez DIMASSI and Anh Tuan DUONG}
\address{Mouez Dimassi, IMB ( UMR CNRS 5251), Universit\'e de Bordeaux 1, 33405 Talence, France}
\email{mouez.dimassi@math.u-bordeaux1.fr}
\address{Anh Tuan Duong, LAGA (UMR CNRS 7539), Univ. Paris 13, F-93430 Villetaneuse, France}
\email{duongat@math.univ-paris13.fr}

\keywords{Magnetic Schr\"odinger operators, asymptotic trace formula, eigenvalues distribution}
\subjclass[2010]{ 81Q10, 35J10, 35P20, 35C20, 47F05}
\maketitle
\begin{abstract}
In this paper,  we consider the $2D$-Schr\"odinger operator with constant magnetic field $H(V)=(D_x-By)^2+D_y^2+V_h(x,y)$, where $V$ tends to  zero at infinity and $h$ is a small positive  parameter.
We will be concerned with  two cases: the semi-classical limit regime  $V_h(x,y)=V(h x,h y)$, 
and the large coupling constant limit  case $V_h(x,y)=h^{-\delta} V(x,y)$. 
We obtain  a complete asymptotic expansion in powers of $h^2$ of 
${\rm tr}(\Phi(H(V),h))$, where $\Phi(\cdot,h)\in C^\infty_0(\mathbb R;\mathbb R)$.
We also give a Weyl type asymptotics  formula with optimal remainder estimate  of  the counting function of eigenvalues of $H(V)$.
\end{abstract}
\section{Introduction}
Let $H_0=(D_x-By)^2+D_y^2$ be the $2D$-Schr\"{o}dinger operator with constant magnetic field $B>0$.
Here $D_\nu=\frac{1}{i}\partial_\nu$. It is well known that the operator $H_0$ is essentially self-adjoint  on 
 $C_0^\infty(\mathbb{R}^{2})$ and  its spectrum consists of eigenvalues of infinite multiplicity 
(called Landau levels, see, e.g., \cite{AHS78}). We denote by $\sigma(H_0)$ (resp. $\sigma_{\rm ess}(H_0)$) the  spectrum (resp. the essential spectrum) of the operator $H_0$. Then, 
$$ \sigma(H_0)= \sigma_{\rm ess}(H_0)=\bigcup_{n=0}^\infty\{(2n+1)B\}.$$

 Let  $V\in C^\infty(\mathbb{R}^{2};\mathbb{R})$ and assume that $V$ is bounded with all its 
 derivatives  and satisfies
  \begin{equation}\label{e0}
  \lim\limits_{|(x,y)|\to\infty}V(x,y)=0.
  \end{equation}
We now consider the perturbed Schr\"{o}dinger operator
\begin{equation}\label{e1}
H(V)=H_0+\ V_h(x,y),
\end{equation}
 where  $V_h$ is  a potential depending on a semi-classical parameter $h>0$, and is of the form 
 $V_h(x,y)=V(hx,hy)$ or $V_h(x,y)=h^{-\delta}V(x,y)$, ($\delta>0$). 
   Using the Kato-Rellich theorem and the Weyl criterion, one sees that $H(V)$ is essentially self-adjoint on 
$C_0^\infty(\mathbb{R}^{2})$ and 
$$\sigma_{\rm ess}(H(V))=\sigma_{\rm ess}(H_0)=\bigcup_{n=0}^\infty\{(2n+1)B\}.$$

The spectral properties of the $2D$-Schr\"odinger operator with constant magnetic field 
$H(V)$  have been intensively
studied in the last ten years. In the case of perturbations, the Landau levels $\Lambda_n=(2n+1)B$
become accumulation points of the eigenvalues  and the asymptotics of the
function counting the number of the eigenvalues lying in a neighborhood of $\Lambda_n$ have been
examined by many authors in different aspects. For recent results, the reader may consult
 \cite{R98, Di01, RW02, MR03, KR09,DiPe10_1, BMR11} and the references therein. 

The asymptotics  with precise remainder estimate for the counting spectral function of the operator 
 $H(h):=H_0+V(hx,hy)$ have been obtained by V. Ivrii  \cite{Ivrii98}. 
In fact, he constructs a micro-local canonical form for  $H(h)$, which leads to the sharp remainder estimates.

However, there are only a few works treating the  case of the large coupling constant
limit (i.e., $V_h(x,y)=h^{-\delta}V(x,y)$) (see  \cite{MR94,R91, R93}). In this case, the asymptotic behavior of the counting spectral function   depends both on the sign
of the perturbation and on its decay properties at infinity. 
In \cite{R93},  G. Raikov obtained only the main asymptotic term  of the counting spectral function
as $h\searrow 0$.

The method used in \cite{R93}  is of variational nature. 
 By
this method one can find the main term in the asymptotics of the counting spectral function  with a weaker
assumption on the perturbation $V$. However, it is  quite difficult to establish with
these techniques an asymptotic formula involving sharp remainder estimates.  

For both the semi-classical and large coupling constant limit,  we  give a complete asymptotic expansion 
 of  the trace of $\Phi(H(V),h) $ in powers of $h^2$.
We also establish a Weyl-type asymptotic formula  with optimal remainder estimate
 for the counting function  of eigenvalues of $H(V)$.
The remainder estimate  in Corollary \ref{t1610} and Corollary \ref{t7} is ${\mathcal O}(1)$,   so it is
better than in the standard case (without magnetic field, see e.g., \cite{Di06}).

To prove our  results, we show that the spectral study of $H(h)$ near some energy level
$z$ can be  reduced to the study of an $h^2$-pseudo-differential operator $E_{-+}(z)$
called {\em the effective Hamiltonian}.  Our results are still true for the case of dimension $2d$ with $d\geq 1$.   For the transparency of the presentation, we shall mainly be concerned with the two-dimensional case. 

The paper is organized as follows: In the next section  we state the
assumptions and the results precisely, and we give an outline of the proofs. 
In Section \ref{s3} we reduce the spectral study of $H(V)$ to the one of a system of $h^2$-pseudo-differential operators $E_{-+}(z)$. In Section \ref{s4},  we establish
 a trace formula involving the effective Hamiltonian $E_{-+}(z)$, and we prove  the results concerning the semi-classical  case.
Finally, Section \ref{s5} is devoted to the proofs of the results concerning the large coupling constant limit case.


\section{Formulations of main results}
\setcounter{equation}{0}

\subsection{Semi-classical case}  In this section  we will be concerned with  the semi-classical magnetic Schr\"odinger operator
$$H(h)=H_0+V(hx,hy),$$
where $V$ satisfies \eqref{e0}. By choosing $B={\rm constant}$,  we may actually assume that $B=1$.

Fix two  real numbers $a$ and $b$ such that $[a,b]\subset \mathbb{R}\setminus \sigma_{\rm ess}(H(h))$. We define

\begin{equation}\label{e24051}
l_0:={\rm min}\left\{q\in {\mathbb N}; V^{-1}([a-(2q+1),b-(2q+1)])\not=\emptyset\right\},\,\, 
\end{equation}
$$ l:={\rm sup}\left\{q\in {\mathbb N}; V^{-1}([a-(2q+1),b-(2q+1)])\not=\emptyset\right\}.$$

 We will give an asymptotic expansion   in powers of 
$h^2$ of ${\rm tr}(f(H(h),h))$ in  the two 
following cases:\\
a)
$f(x,h)=f(x)$, where $f\in C^\infty_0(\left(
a,b\right);{\mathbb R})$.\\
b)
$f(x,h)=f(x)\widehat\theta(\frac{x-\tau}{h^2})$, where $f,\theta\in
 C^\infty_0({\mathbb R};\mathbb R)$, $\tau \in \mathbb R$, and 
$\widehat\theta$ is the Fourier transform of $\theta$.\\
As a consequence, we get a sharp remainder estimate for the counting
spectral function of ${H(h)}$ when $h\searrow 0$.
Let us state the results precisely.

\begin{theorem}\label{t3}
Assume \eqref{e0},  and let  $f\in C_0^\infty((a,b);\mathbb{R})$.
There exists a sequence of real numbers $(\alpha_j(f))_{j\in {\mathbb N}}$, such that
\begin{equation}\label{e29055}
{\rm tr}(f(H(h)))\sim\sum_{k=0}^\infty \alpha_k(f)h^{2(k-1)},
\end{equation}
where 
\begin{equation}
\alpha_0(f)=\sum_{j=l_0}^{l}\frac{1}{2\pi}\iint f((2j+1)+V(x,y))dxdy.
\end{equation}
\end{theorem}

Let  $\theta\in C^\infty_0({\mathbb R})$, and let $\epsilon$ be a
positive  constant.
Set
$$\breve\theta(\tau)={1\over 2\pi }\int e^{it\tau}\theta(t)dt, 
\quad
\breve{\theta}_\epsilon(t)={1\over 
\epsilon}\breve{\theta}({t\over\epsilon}).$$
In the sequel we shall say that $\lambda$ is not a critical value of
$V$ if and only if $V(X)=\lambda$ for some $X \in \mathbb R ^{2}$ implies
$\nabla_X V(X) \neq 0$.

\begin{theorem}\label{t291102}  Fix $\mu\in\mathbb{R}\setminus\sigma_{\rm ess}(H(h))$ which is not a critical
value of $(2j+1+V)$, for $j=l_0,...,l$. Let $f\in C^\infty_0(\left(\mu-\epsilon,\mu+
\epsilon\right);\mathbb R)$ and
$\theta\in C^\infty_0(\left(-{1\over C},{1\over C}\right); {\mathbb R})$, with $\theta=1$ near $0$. 
Then there exist  $\epsilon >0$, $C>0$ 
and a functional sequence  $c_j \in C^\infty({\mathbb R}; \mathbb R)$, $j \in {\mathbb N}$,
such that for all $M, N\in {\mathbb  N}$,  we have 

\begin{equation}\label{e290561}
 {\rm tr}\left(f(H(h))\breve\theta_{h^2}(t-H(h))\right) =\sum_{k=0}^Mc_k(t)h^{2(k-1)}+\mathcal{O}\left(
\frac{h^{2M}}{\langle t\rangle^N}\right)
\end{equation}
uniformly in $t\in\mathbb{R}$,
where
\begin{equation}\label{e061201}
 c_0(t)=\frac{1}{2\pi}f(t)\sum_{j=l_0}^l\int_{\{(x,y)\in\mathbb{R}^2;\;2j+1+V(x,y)=t\} }\frac{d{S_t}}{|\nabla V(x,y)|}.
\end{equation}

\end{theorem}

\begin{corollary}\label{t5}
In addition to the hypotheses of Theorem \ref{t3} suppose that $a $  and $b$ are not critical values of $((2j+1)+V)$ for all $j=l_0,...,l$.
Let ${\mathcal N}_h([a,b])$ be the number of eigenvalues of $H(h)$ in the interval $[a,b]$ counted with their multiplicities. Then we have 
\begin{equation}
{\mathcal N}_h([a,b])=h^{-2} C_0+\mathcal{O}(1),h\searrow 0,
\end{equation}
where 
\begin{equation}
 C_0=\frac{1}{2\pi}\sum_{j=l_0}^{l}{\rm Vol}\left(V^{-1}([a-(2j+1),b-(2j+1)])\right).
\end{equation}
\end{corollary}

\subsection{Large coupling constant limit case.}

We  apply the above results to the Schr\"{o}dinger
 operator with constant magnetic field in the large coupling constant limit case. 
More precisely, consider 
\begin{equation}\label{e10}
H_{\lambda}=(D_{x}-y)^2+D_{y}^2+\lambda V(x,y).
\end{equation}
Here $\lambda$ is a large constant, and
 the electric potential $V$ is assumed to be strictly positive. Let $X:=(x,y)\in\mathbb{R}^2$.
We suppose in addition that for all $N\in\mathbb{N}$,
\begin{equation}\label{e11}
V(X)=\sum_{j=0}^{N-1}\omega_{2j}\left(\frac{X}{\left|X\right|}\right)|X|^{-\delta-2j}+r_{2N}(X),\mbox{ for }
 \left|X\right|\geq 1,
\end{equation}
where
\begin{itemize}
 \item $\omega_0\in C^\infty(\mathbb{S}^{1};(0,+\infty))$, 
$\omega_{2j}\in C^\infty(\mathbb{S}^{1};\mathbb{R})$, $j\geq 1$. Here 
$\mathbb{S}^1$ denotes the unit circle.
\item $\delta$ is some positive constant,
\item $|\partial_X^\beta r_{2N}(X)|\leq C_\beta(1+|X|)^{-|\beta|-\delta-2N},\;\forall\beta\in\mathbb{N}^2$.
\end{itemize}

Since $V$ is positive, it follows that $\sigma(H_\lambda)\subset[1,+\infty)$. 
Fix two  real numbers $a$ and $b$ such that $a>1$ and 
$[a,b]\subset \mathbb{R}\setminus \sigma_{\rm ess}(H_\lambda)$.  
Since $\sigma_{\rm ess}(H_\lambda)=\cup_{j=0}^\infty\{(2j+1)\}$,
there  exists $q\in{\mathbb N}$ such that $2q+1<a<b<2q+3$.  
The following results are consequences of Theorem \ref{t3}, Theorem \ref{t291102} and Corollary \ref{t5}.
\begin{theorem}\label{t1610}
Assume \eqref{e11},  and let  $f\in C_0^\infty((a,b);\mathbb{R})$.
There exists a sequence of real numbers $(b_j(f))_{j\in {\mathbb N}}$, such that
\begin{equation}\label{e121101}
{\rm tr}(f(H_{\lambda}))\sim\lambda^{\frac{2}{\delta}}
\sum_{k=0}^\infty b_k(f)\lambda^{-\frac{2k}{\delta}},\;\lambda\nearrow+\infty
\end{equation}
where
\begin{equation}
b_0(f)=\frac 1{2\pi\delta}\int_0^{2\pi} (\omega_0(\cos \theta,\sin \theta))^{\frac 2{\delta}} d\theta \sum_{j=0}^q\int  f(u)(u-(2j+1))^{-1-\frac 2{\delta}} du.
\end{equation}
\end{theorem}

\begin{theorem} \label{t291101}
 Let $f\in C^\infty_0(\left(a-\epsilon,b+
\epsilon\right);\mathbb R)$ and
$\theta\in C^\infty_0(\left(-{1\over C},{1\over C}\right); {\mathbb R})$, with $\theta=1$ near $0$. 
Then there exist  $\epsilon >0$, $C>0$ 
and a functional sequence  $c_j \in C^\infty({\mathbb R}; \mathbb R)$, $j \in {\mathbb N}$,
such that for all $M, N\in {\mathbb  N}$,  we have

\begin{equation}\label{e29056}
 {\rm tr}\left(f(H_\lambda)\breve\theta_{\lambda^{-\frac{2}{\delta}}}(t-H_\lambda)\right) =\lambda^{\frac{2}{\delta}}\sum_{k=0}^{M}c_k(t)\lambda^{-\frac{2k}{\delta} }+\mathcal{O}\left(
\frac{\lambda^{-\frac{2M}{\delta}}}{\langle t\rangle^N}\right)
\end{equation}
uniformly in $t\in\mathbb{R}$,
where
\begin{equation}
 c_0(t)=\frac{1}{2\pi}f(t)\sum_{j=0}^q\int_{\{X\in\mathbb{R}^2;\;2j+1+W(X)=t\} }\frac{d{S_t}}{|\nabla_X W(X)|}.
\end{equation}
Here  $W(X)=\omega_0(\frac{X}{\vert X\vert})\vert X\vert^{-\delta}$.

\end{theorem}

\begin{corollary}\label{t7}
Let ${\mathcal N}_\lambda([a,b])$ be the number of eigenvalues of $H_{\lambda}$ in the interval 
$[a,b]$ counted with their multiplicities.  We have
$$ {\mathcal N}_\lambda([a,b])=\lambda^{\frac{2}{\delta}} D_0+\mathcal{O}(1),\;\lambda\nearrow +\infty,$$
where $$ D_0=\frac{1}{4\pi}\sum_{j=0}^{q}\left((a-2j-1)^{-\frac{2}{\delta}}-(b-2j-1)^{-\frac{2}{\delta}}\right) \int_0^{2\pi} \left(\omega_0(\cos \theta,\sin \theta)\right)^{\frac{2}{\delta}}d\theta.$$
\end{corollary}

\subsection{Outline of the proofs}

The purpose  of this subsection is to provide a broad  outline of the proofs. By a change of variable on the phase space, the operator $H(h)$ is unitarily equivalent to
$$P(h):=P_0+V^w(h)=
-\frac{\partial^2}{\partial y^2} + y^2+  V^w(x+hD_y,hy+h^2D_x), X=(x,y)\in {\mathbb R}^2.$$
Let $\Pi=1_{[c,d]}(P_0)$ be the spectral projector of the harmonic oscillator on the interval $[c,d]=\left[a-\Vert V\Vert_{L^\infty(\mathbb{R}^2)}, b+\Vert V\Vert_{L^\infty(\mathbb{R}^2)}\right]$.  Using the explicit expression of $\Pi$ we will reduce the spectral study of $(P-z)$ for $z\in [a,b]+i[-1,1]$ to the study of a system of $h^2$-pseudo-differential operator, $E_{-+}(z)$  depending  only on $x$ (see Remark \ref{r1} and Corollary \ref{c201201}).
In particular, modulo ${\mathcal O}(h^\infty)$, we are reduced   to proving  Theorem \ref{t3} and Theorem \ref{t291102} for a system of $h^2$-pseudo-differential operator  (see Proposition \ref{p1}). 
Thus,   \eqref{e29055} and \eqref{e290561} follows easily  from Theorem 1.8 in \cite{Di98} (see also \cite{Di01}). Corollary \ref{t5} is a simple consequence of Theorem \ref{t3}, Theorem \ref{t291102} and a Tauberian-argument.

To deal with the large coupling constant limit case, we note that  for all $M>0$ and $\lambda$ 
 large enough, we have
 $$\left\{(x, y,\eta,\xi)\in {\mathbb R}^4; \vert (x,y)\vert<M,\,\, (\xi-y)^2+\eta^2+\lambda V(x,y)\in [a,b]\right\}=\emptyset.$$
 Thus, on the symbolic level, only the behavior of $V(x,y)$ at infinity contributes to the asymptotic behavior of  the left hand sides of \eqref{e121101} and \eqref{e29056}. 
 Since, for $\vert X\vert $ large enough, $\lambda V(X)=\varphi_0(hX)+\varphi_2(hX)h^2+\cdots+ 
\varphi_{2j}(hX)h^{2j}+\cdots$ with $h=\lambda^{-\frac 1{\delta}}$ and $\varphi_0(X)=\omega_0(\frac{X}{|X|})|X|^{-\delta}$, 
 Theorem \ref{t1610} (resp. Theorem \ref{t291101}) follows from Theorem \ref{t3} (resp. Theorem \ref{t291102}).

\section{The effective Hamiltonian}\label{s3}

\subsection {Classes of symbols}
\setcounter{equation}{0}

Let $M_n(\mathbb{C})$ be the space of complex square matrices of order $n$. We recall  the standard class of  semi-classical matrix-valued symbols on $T^*{\mathbb R}^d={\mathbb R}^{2d}$:
 $$
 S^m({\mathbb R}^{2d};M_n({\mathbb C}))=\left\{a\in { C}^\infty({\mathbb R}^{2d}\times (0,1]; M_n({\mathbb C})) ;\,
\|  \partial_x^\alpha\partial_\xi^\beta a(x,\xi;h)\|_{M_n({\mathbb C})}\leq  C_{\alpha,\beta} h^{-m}\right\}.$$
We note that the symbols are tempered as $h\searrow 0$.  The more general class $S^m_\delta({\mathbb R}^{2d};M_n({\mathbb C}))$, where the right hand side in the above estimate is replaced by $C_{\alpha,\beta} h^{-m-\delta(\vert \alpha\vert+\vert \beta\vert)}$, has nice  quantization properties as long as $0\leq \delta\leq \frac{1}{2}$ (we refer to \cite[Chapter 7]{DiSj99}).

For $h$-dependent symbol $a\in S_\delta^m({\mathbb R}^{2d};M_n({\mathbb C}))$, we say that $a$ has an asymptotic  expansion in powers of $ h$ in $ S_\delta^m(\mathbb{R}^{2d};M_n(\mathbb{C}))$ and we write 
 $$a\sim\sum\limits_{j\geq 0}a_jh^j,$$
if there exists a sequence of symbols
$a_j(x,\xi)\in S_\delta^m(\mathbb{R}^{2d};M_n(\mathbb{C}))$ such that for all $N\in {\mathbb N}$, we have
$$a-\sum_{j=0}^Na_jh^j\in S_\delta^{m-N-1}(\mathbb{R}^{2d};M_n(\mathbb{C})).$$

In the special case when $m=\delta=0$ (resp. $m=\delta=0, n=1)$, we will write
 $S^0(\mathbb{R}^{2d};M_n(\mathbb{C}))$  (resp.  $S^0(\mathbb{R}^{2d})$) instead  of  $S^0_0(\mathbb{R}^{2d};M_n(\mathbb{C}))$
( resp. $S_0^0(\mathbb{R}^{2d}; M_1({\mathbb C})$).

We will use the
standard Weyl quantization of symbols. More precisely, if
$a\in S_\delta^m(\mathbb{R}^{2d};M_n(\mathbb{C}))$, then $a^w(x,hD_x;h)$
 is the operator defined by

$$a^w(x,hD_x;h)u(x)=(2\pi h)^{-d}\iint e^{\frac{i(x-y,\xi)}{h}}a\left(\frac{x+y}{2},\xi;h\right)u(y)dyd\xi,\;u\in\mathcal{S}(\mathbb{R}^d;\mathbb{C}^n).$$

In order to prove our main results, we shall recall some well-known  
results 

\begin{proposition}\label{p291102} {\bf (Composition formula)}
Let $a_i\in S^{m_i}_\delta({\mathbb R}^{2d};M_n( {\mathbb C}))$, $i=1,2$,  $\delta\in
\left. \lbrack 0,{1\over 2}\right)$. Then
$b^w(y,h D_y;h)=a_1^w(y,h D_y)\circ a^w_2(y,h D_y)$ is an 
$h$-pseudo-differential operator, and
$$
b(y,\eta;h)\sim\sum_{j=0}^\infty b_j(y,\eta)
h^j, \,\,\hbox {\rm in } S^{m_1+m_2}_\delta({\mathbb R}^{2d}; M_n({\mathbb C})).
$$
\end{proposition}
\begin{proposition}\label{p291103}{\bf (Beals characterization)}
 Let $A=A_h: {\mathcal S}({\mathbb R}^d; \mathbb C^n)\to {\mathcal S}'({\mathbb R}^d; \mathbb C^n)$, $0<h\leq 1$.
 The following two statements are equivalent:\\
 (1) $A=a^w(x,h D_x;h)$, for some $a=a(x,\xi ;h)\in S^0({\mathbb R}^{2d}; M_n(\mathbb C))$.\\
(2) For every $N\in {\mathbb N}$ and for every sequence $l_1(x,\xi )$,\dots
, $l_N(x,\xi)$ of linear forms on ${\mathbb R}^{2d}$, the operator
${\rm ad}_{l_1^w(x,h D_x)}\circ \dots \circ {\rm ad}_{l_N^w(x,h
D_x)}A_h$ belongs to ${\mathcal L}(L^2,L^2)$ and is of norm
${\mathcal O}(h^N)$ in that space.
Here, ${\rm ad}_AB:=\lbrack A,B\rbrack=AB-BA$.
\end{proposition}
\begin{proposition}\label{p291104}{\bf ($L^2-$ boundedness)}
 Let $a=a(x,\xi ;h)\in S^0_\delta({\mathbb R}^{2d}; M_n(\mathbb C))$, $0\leq \delta\leq 1/2$. Then $a^w(x,hD_x;h)$ is bounded :  $L^2({\mathbb R}^d; {\mathbb C}^n) \rightarrow
 L^2({\mathbb R}^d; {\mathbb C}^n)$, 
 and there is a constant $C$ independent of $h$ such that for $0<h\leq 1$;
 $$\Vert a^w(x,hD_x;h)\Vert\leq C.$$
\end{proposition}

\subsection {Reduction to a semi-classical problem}
\setcounter{equation}{0}


Here, we shall  make use of a strong field reduction onto the $j$th
eigenfunction of the harmonic oscillator, $j=l_0\cdots l$, and  a well-posed Grushin problem
for $H(h)$. We show that the spectral study of $H(h)$ near some energy level
$z$ can be  reduced to the study of an $h^2$-pseudo-differential operator $E_{-+}(z)$
called {\em the effective Hamiltonian}.  
Without any loss of generality we may 
assume that $l_0=1$. \\
\begin{lemma} 
There exists a unitary operator
$\widetilde{W}: L^2({\mathbb R}^2) \to  L^2({\mathbb R}^2)$ such that
$$
P(h)=
\widetilde{\mathcal W} H(h) \widetilde{\mathcal W}^*
$$
where $P(h): =  P_{0}
+  V^w(h)$,  $P_{0}: =
-\frac{\partial^2}{\partial y^2} + y^2$ and
${ V}^w(h): =  V^w(x+hD_y,hy+h^2D_x) $.
\end{lemma}
\begin{proof}
The linear symplectic mapping
$$
\widetilde{S}:\;\mathbb{R}^4\to\mathbb{R}^4\mbox{  given by }(x,y,\xi,\eta)\mapsto\left(\frac{1}{h} x+\eta,y+h\xi,h\xi,\eta\right), 
$$
maps the Weyl symbol of the operator $H(h)$ into the
Weyl symbol of the operator $P(h)$. By Theorem A.2  in \cite[Chapter 7]{DiSj99},
there exists a unitary operator $\widetilde{W}: L^2({\mathbb R}^2) \to  L^2({\mathbb R}^2)$
 associated to  $\widetilde{S}$ such that $
P(h)=
\widetilde{\mathcal W} H(h) \widetilde{\mathcal W}^*
$. 
\end{proof}


Introduce the operator $R^-_j: L^2({\mathbb R}_x) \rightarrow  
L^2({\mathbb R}^2_{x,y})$ by
$$
(R^-_jv)(x,y)=\phi_j(y) v(x),
$$
where $\phi_j$ is the $j$th normalized eigenfunction of the  harmonic 
oscillator. 
Further, the operator $R_j^+: L^2({\mathbb R}^2_{x,y}) \rightarrow
L^2({\mathbb R}_x)$ is defined by
$$(R_j^+u)(x)=\int  \phi_j(y) u(x,y)dy.$$
Notice that $R_j^+$ is the adjoint of $R_j^-$.
An easy computation shows that 
$R_j^+R_j^-=I_{L^2({\mathbb R}_x)}$ and $R_j^-R_j^+=\Pi_j$, where\\
$$\Pi_j:\;L^2({\mathbb R}^{2})\to L^2({\mathbb R}^{2}),\;
v(x,y)\mapsto \int v(x,t)\phi_j(t)dt\phi_j(y).$$

Define  $\Pi=\sum\limits_{j=1}^{l}\Pi_j$.
\begin{lemma}\label{l1}
Let $\Omega:=\{z\in\mathbb{C};\;{\rm Re}z\in [a,b],\,\, \vert {\rm Im} z\vert <1\}$. The operator
$$(I-\Pi)P(h)(I-\Pi)-z:\;(I-\Pi)L^2({\mathbb R}^{2})\to (I-\Pi)L^2({\mathbb R}^{2})$$
is uniformly  invertible for $z\in \Omega$.
\end{lemma}
\begin{proof}
It follows from the definition of $\Pi$ that  $\sigma((I-\Pi)P_0(I-\Pi))=\bigcup_{k\in \mathbb{N}\setminus\{1,...,l\}}\{2k+1\}$. Hence 
$$\sigma((I-\Pi)P(h)(I-\Pi))\subset\bigcup_{k\in \mathbb{N}\setminus\{1,...,l\}}
\left[2k+1-\|V\|_{L^\infty(\mathbb{R}^2)},2k+1+\|V\|_{L^\infty(\mathbb{R}^2)}\right],$$
 which implies  
 $$\sigma((I-\Pi)P(h)(I-\Pi))\cap [a,b]=\emptyset.$$ Consequently,
$$\|(I-\Pi)P(h)(I-\Pi)-z\|\geq {\rm dist}\left([a,b],\sigma((I-\Pi)P(h)(I-\Pi))\right)>0$$
uniformly for  $z\in \Omega$. Thus, we obtain $$(I-\Pi)P(h)(I-\Pi)-z:\;(I-\Pi)L^2({\mathbb R}^{2})\to (I-\Pi)L^2({\mathbb R}^{2})$$
is  uniformly  invertible for  $z\in\Omega$.
\end{proof}
For $z\in\Omega$, we put
$$\mathcal{P}(z)=\begin{pmatrix}{}
(P(h)-z)&R_1^-&...&R_l^-\\
R_1^+&0&...&0\\
.&.&...&.\\
.&.&...&.\\
.&.&...&.\\
R_l^+&0&...&0
\end{pmatrix}\mbox{ and }\widetilde{\mathcal{E}}(z)=\begin{pmatrix}{}
R(z)&R_1^-&...&R_l^-\\
R_1^+&A_1&...&0\\
.&.&...&.\\
.&.&...&.\\
.&.&...&.\\
R_l^+&0&...&A_l
\end{pmatrix},$$
where $A_j=z-(2j+1)-R_j^+V ^w(h)R_j^-,j=1,...,l$ and $R(z)=((I-\Pi)P(h)(I-\Pi)-z)^{-1}(I-\Pi)$.

Let $\mathcal{E}_1(z):=\mathcal{P}(z)\widetilde{\mathcal{E}}(z)=(a_{k,j})_{k,j=1}^{l+1}$. In the next step we will compute explicitly $a_{k,j}$.

 Using the fact that $\Pi R(z)=0$ as well as the fact that  $\Pi$ commutes with $P_0$, we deduce that  $(P(h)-z)R(z)=(I-\Pi)+[\Pi,V^w(h)]R(z)$.
Consequently, 
\begin{equation}
a_{1,1}=(P(h)-z)R(z)+\sum_{j=1}^lR_j^-R_j^+=I+[\Pi,V^w(h)]R(z).
\end{equation}
Next, from the definition of $A_1$ and the fact that $P_0R_1^-=3R_1^-$ (we recall that $l_0=1$), one has
\begin{align*}
a_{1,2}&=(P(h)-z)R_1^-+R_1^-A_1\\ 
&=-(z-3)R_1^-+V^w(h)R_1^-
+R_1^-(z-3)-\Pi_1V^w(h)R_1^-\\
&=V^w(h)R_1^--\Pi_1V^w(h)R_1^-\\
&=[V^w(h),\Pi_1]R_1^-.
\end{align*}
Similarly, $a_{1,j}=[V^w(h),\Pi_{j-1}]R_{j-1}^-, j\geq 3$.

Since $R_1^+(1-\Pi_1)=R_1^+-R_1^+R_1^-R_1^+=0$ and $R_1^+\Pi_j=R_1^+R_j^-R_j^+=0$
 for $j\not=1$,  it follows that  $a_{2,1}=R_1^+R(z)=0$.  Evidently, $a_{2,2}=R_1^+R_1^-=I_{L^2(\mathbb{R})}$ 
and $a_{2,j}=R_1^+R_j^-=0$ for $j\geq 3$. The same arguments as above show that   $a_{k,j}=\delta_{j,k} I_{L^2(\mathbb{R})}$ for all $k\geq 3$.
Summing up we have proved
$$\mathcal{E}_1(z)=\mathcal{P}(z)\widetilde{\mathcal{E}}(z)=\begin{pmatrix}{}
I+[\Pi,V^w(h)]R(z)&[V^w(h),\Pi_1]R_1^-&...&[V^w(h),\Pi_l]R_l^-\\
0&I_{L^2(\mathbb{R})}&...&0\\
.&.&...&.\\
.&.&...&.\\
.&.&...&.\\
0&0&...&I_{L^2(\mathbb{R})}
\end{pmatrix},$$
Let $f_j\in C_0^\infty(\mathbb{R})$, $f_j=1$ near $2j+1$ and supp$f_j\subset [2j,2j+2]$. By the spectral theorem we have  $\Pi_j=f_j(D_y^2+y^2)$.
On the other hand, the functional calculus of  pseudo-differential operators shows that  $\Pi_j=f_j(D_y^2+y^2)=B^w(y,D_y)$
with $B(y,\eta)={\mathcal O}(\langle y\rangle^{-\infty}\langle \eta\rangle^{-\infty})$.
 
The composition formula of pseudo-differential operators (Proposition \ref{p291102}) gives
\begin{equation}\label{e2}
[V^w(h),\Pi_j]= \sum_{k=1}^Nb_{k,j}^w(x,h^2D_x)c_{k,j}^w(y,D_y)h^k+\mathcal{O}(h^{N+1}),\;\forall N\in\mathbb{N}, 
\end{equation}
where $b_{k,j}, c_{k,j}\in S^0(\mathbb{R}^2)$. This together with the Calderon-Vaillancourt theorem (Proposition \ref{p291104}) yields   $[V^w(h),\Pi_j]=\mathcal{O}(h)$ in $\mathcal{L}(L^2(\mathbb{R}^2))$.
Therefore,  for  $h$ is sufficiently small, $\mathcal{E}_1(z)$ is  uniformly invertible for $z\in \Omega$, and
$$\mathcal{E}_1(z)^{-1}=\begin{pmatrix}{}
a(z)&-a(z)[V^w(h),\Pi_1]R_1^-&...&-a(z)[V^w(h),\Pi_l]R_l^-\\
0&I_{L^2(\mathbb{R})}&...&0\\
.&.&...&.\\
.&.&...&.\\
.&.&...&.\\
0&0&...&I_{L^2(\mathbb{R})}
\end{pmatrix},$$
where $a(z)=(I+[\Pi,V^w(h)]R(z))^{-1}$. 
Using the explicit expressions of $\widetilde{\mathcal{E}}(z)$ 
and $\mathcal{E}_1(z)^{-1}$,  we get

$$\mathcal{E}(z):=\widetilde{\mathcal{E}}(z)\mathcal{E}_1(z)^{-1}$$
$$=\begin{pmatrix}{}
R(z)a(z)&-R(z)a(z)[V^w(h),\Pi_1]R_1^-+R_1^-&...&-R(z)a(z)[V^w(h),\Pi_l]R_l^-+R_l^-\\
R_1^+a(z)&A_1-R_1^+a(z)[V^w(h),\Pi_1]R_1^-&...&-R_1^+a(z)[V^w(h),\Pi_l]R_l^-\\
.&.&...&.\\
.&.&...&.\\
.&.&...&.\\
R_l^+a(z)&-R_l^+a(z)[V^w(h),\Pi_1]R_1^-&...&A_l-R_l^+a(z)[V^w(h),\Pi_l]R_l^-
\end{pmatrix}.$$
Thus, we have proved the following theorem.
\begin{theorem}\label{t1}
Let $\Omega$ be as in Lemma \ref{l1}. Then $\mathcal{P}(z)$ is uniformly invertible  for $z\in\Omega$
 with  inverse $\mathcal{E}(z)$. In addition, $\mathcal{E}(z)$ is holomorphic in $z\in\Omega$.
\end{theorem}

From now on, we write $\mathcal{E}(z)=(B_{k,j})_{k,j=1}^{l+1}=\begin{pmatrix}E(z)&E_+(z)\\E_-(z)&E_{-+}(z)\end{pmatrix}$, 
where $E_{-+}(z)=(B_{k,j})_{k,j=2}^{l+1}$, $E(z)=R(z)a(z)$, $E_-(z)=\begin{pmatrix} R_1^+a(z)\\
.\\
.\\
.\\
R_l^+a(z)\end{pmatrix}$, and 
$$E_+(z)=\begin{pmatrix}-R(z)a(z)[V^w(h),\Pi_1]R_1^-+R_1^-&...&-R(z)a(z)[V^w(h),\Pi_l]R_l^-+R_l^-\end{pmatrix}.$$

\begin{remark}\label{r1}
 The following formulas are  consequences of the fact that ${\mathcal E}(z)$ is the inverse of ${\mathcal P}(z)$ as well as the fact that $R_j^\pm$ are independent of $z$ (see \cite{Di01,HeSj89}): 

 \begin{equation}\label{e1406}
  (z-P(h))^{-1}=-E(z)+E_+(z)(E_{-+}(z))^{-1}E_-(z),\;z\in\rho(P(h)),
  \end{equation}
 \begin{equation}\label{e1404}
\partial_z E_{-+}(z)=E_-(z)E_+(z).
 \end{equation}

\end{remark}
In what follows, the explicit formulae  for $E(z)$ and $ E_\pm(z)$ 
are not needed. We just indicate that they are holomorphic  in $z$. 
In the remainder of this section, we will prove that the symbol of 
the operator $E_{-+}(z)$ is in $S^0({\mathbb R}^{2};M_l(\mathbb C))$, and has a complete asymptotic expansion in powers of $h$. Moreover, we will give explicitly the principal  term.


\begin{proposition}\label{p3}
For $1\leq k,j\leq l$, the operators  $R_j^+V^w(h)R_j^-$ and  $R_k^+a(z)[V^w(h),\Pi_j]R_j^-$ are 
$h^2-$pseudo-differential operators with bounded symbols. Moreover, there exist  $v_{j,n}, b_{k,j,n}\in S^0(\mathbb{R}^{2})$, $n=1,2,..$,  such that
\begin{align}
\label{e20071}
R_j^+V^w(h)R_j^- &=\sum_{n=0}^Nh^{2n}v^w_{j,n}(x,h^2D_x)+\mathcal{O}\left(h^{2(N+1)}\right), \\
\label{e26051} R_k^+a(z)[V^w(h),\Pi_j]R_j^-&=\sum_{n=1}^Nb^w_{k,j,n}(x,h^2D_x,z)h^n+ \mathcal{O}\left(h^{N+1}\right), \mbox{ for } k\not=j,\\
\label{e22071}R_j^+a(z)[V^w(h),\Pi_j]R_j^-&=\sum_{n=1}^Nb^w_{j,j,2n}(x,h^2D_x,z)h^{2n}+ \mathcal{O}\left(h^{2(N+1)}\right),\; \forall N\in\mathbb{N},
\end{align}
Here
$$ v_{j,0}(x,\xi)=V(x,\xi),\; \;j=1,...,l.$$
\end{proposition}
\begin{proof}

The proofs of \eqref{e20071},  \eqref{e26051} and   \eqref{e22071} are quite similar, and are based on the  Beal's characterization of $h^2$-pseudo-differential operators (see Proposition \ref{p291103}).  We give only the main ideas of the proof  of \eqref{e20071} and we refer to \cite{Di01,DiSj99,HeSj89} for more details.
Let $Q$ denote  the left hand side of  \eqref{e20071}. Let $l^w(x,h^2D_x)$ be as in Proposition \ref{p291103}. 
Using the fact that $R_j^\pm$ commutes with $l^w(x,h^2D_x)$ as well as the fact that $V^w(h)$ is an $h^2$-pseudo-differential operator on $x$, we deduce from Proposition \ref{p291103} that $Q=
q^w(x,h^2D_x;h)$, with $q\in S^0(\mathbb{R}^2)$.
On the other hand, writing
\begin{equation}
\label{e20000}
V^w(h)=V^w(x,h^2D_x)+ hD_y\left({\partial V\over\partial x}\right)^w(x,h^2D_x)
 +hy \left({\partial V\over\partial y}\right)^w
(x,h^2D_x)+\cdots,
\end{equation}
and  using  Proposition \ref{p291103}, we see that $q(x,\xi;h)$ has an asymptotic 
expansion in powers of $h$. 

Notice that the odd  powers 
of $h$  in \eqref{e20071} and   \eqref{e22071} disappear, due to the special properties of  
the eigenfunctions of the harmonic oscillator (i.e., $\int_{\mathbb R} y^{2j+1} \vert \phi_j(y)\vert^2dy=\int_{\mathbb R} \phi_j(y)\, \partial_y^{2j+1}\phi_j(y)dy=0$). Finally,  since $R^+_jR^-_j=I_{L^2(\mathbb{R})}$, it follows from
 \eqref{e20000} that $v_{j,0}(x,\xi)=V(x,\xi)$.
 
\end{proof}

Let $e_{-+}(x,\xi,z,h)$ denote the symbol of $E_{-+}(z)$. 
The following corollary  follows from the above proposition and the definition of 
 $E_{-+}(z)$.

\begin{corollary}\label{c201201} We have
$$e_{-+}(x,\xi,z,h)\sim\sum_{j=0}^\infty e_{-+}^j(x,\xi,z)h^j,\,\,\ {\rm in}\,\, S^0({\mathbb R}^2;M_l({\mathbb C})),$$
with
$$e_{-+}^0(x,\xi,z)=\Bigl((z-(2j+1)-V(x,\xi))\delta_{i,j}\Bigr)_{1\leq i,j\leq l}.$$
\end{corollary}

\section{Proof of Theorem \ref{t3} and Theorem \ref{t5}}\label{s4}
\subsection{Trace formulae}\label{sub41}
Let $f\in C_0^\infty((a,b);\mathbb{R})$, where $(a,b)\subset \mathbb{R}\setminus\sigma_{\rm ess}(P(h))$, and let $\theta\in C^\infty_0(\mathbb R;\mathbb R)$. 
 Set
  $$\Sigma_j([a,b])=\left\{(x,\xi)\in\mathbb{R}^2;\;2j+1+V(x,\xi)\in[a,b]\right\},j=1,...,l$$
   and
\begin{equation}
 \Sigma_{[a,b]}=\bigcup_{j=1}^l \Sigma_j([a,b]).
\end{equation}
 Let $\widetilde{f}\in C_0^\infty((a,b)+i[-1,1])$    be an almost analytic extension of $f$, i.e., $\widetilde f=f$ on ${\mathbb R}$ and 
$\overline \partial_z\widetilde f$ vanishes on ${\mathbb R}$ to infinite 
order, i.e. ${\overline{\partial }_z\widetilde{f}(z)={\mathcal
O}_N(\vert {\rm Im} \; z\vert ^N)\hbox{ for all }N\in {\mathbb N}.}$
 Then the functional calculus due to Helffer-Sj\"ostrand 
(see e.g. \cite[Chapter 8]{DiSj99}) yields 
\begin{equation}\label{e1409}
f(P(h))=-{1\over \pi}\int
\overline\partial_z
\widetilde f(z)(z-P(h))^{-1}L(dz),
\end{equation}
\begin{equation}\label{e1410}
f(P(h))\breve\theta_{h^2}(t-P(h))=-{1\over \pi}\int
\overline\partial_z
\widetilde f(z) \breve \theta_{h^2}(t-z)(z-P(h))^{-1}L(dz).
\end{equation}
Here $L(dz)=dxdy$ is the Lebesgue measure on the complex plane 
${\mathbb  C}\sim {\mathbb R}^2_{x,y}$.  In the last equality we have used the fact that  $\widetilde f(z) \breve \theta_{h^2}(t-z)$ is an almost analytic extension of $f(x) \breve \theta_{h^2}(t-x)$,
since $z\mapsto \breve \theta_{h^2}(t-z)$ is analytic.

\begin{proposition}\label{p1}
For $h$ small enough, we have
 \begin{equation}\label{e12}
{\rm tr}(f(P(h)))={\rm tr}\left(-\frac{1}{\pi}\int \overline{\partial}_z\widetilde{f}(z)(E_{-+}(z))^{-1}\partial_z E_{-+}(z)L(dz)\chi^w(x,h^2D_x)\right)+\mathcal{O}(h^\infty),
 \end{equation}
  \begin{equation}\label{e1414}
{\rm  tr }   \left(f(P(h))   \breve\theta_{h^2}(t-P(h))\right)=
 \end{equation}
 $$
{\rm tr}\left(-\frac{1}{\pi}\int \overline{\partial}_z\widetilde{f}(z)\breve \theta_{h^2}(t-z)(E_{-+}(z))^{-1}\partial_z E_{-+}(z)L(dz)\chi^w(x,h^2D_x)\right)+\mathcal{O}(h^\infty),
$$
where $\chi\in C_0^\infty({\mathbb{R}}^{2};{\mathbb{R}})$ is equal to one in a neighbourhood of    $\Sigma_{[a,b]}$. 
\end{proposition}
\begin{proof}
Replacing $(z-P(h))^{-1}$ in   \eqref{e1409}  by  the right hand
side of   \eqref{e1406}, and using the fact that $E(z)$  is holomorphic in $z$,
we obtain
\begin{equation}\label{e5}
f({P(h)})=-\frac{1}{\pi}\int \overline{\partial}_z\widetilde{f}(z)E_+(z)(E_{-+}(z))^{-1}E_-(z)L(dz).
\end{equation}

Let $\widetilde{V}\in S^0({\mathbb R}^2)$ be a real-valued function coinciding with
$V$ for large $(x,y)$, and having the property that 
\begin{equation}\label{e1200}
\vert z-(2j+1)-\widetilde{V}(x,y)\vert>c>0, \,\, j=1,2,...,l,
\end{equation}
uniformly in $z\in{\rm supp} \; \widetilde f$,  and $(x,y)\in {\mathbb R}^2$. We recall that for $z\in {\rm supp} \; \widetilde f$, ${\rm Re} z\in (a,b)\subset \mathbb R\setminus
\sigma_{\rm ess } (H(h))=\mathbb R\setminus\cup_{k=0}^\infty \{(2k+1)\}$. Then \eqref{e1200} holds for $\widetilde{V}\in S^0({\mathbb R}^2)$ with $\Vert \widetilde V\Vert$ small enough.

Set $\widetilde E_{-+}(z):=E_{-+}(z)+\left (V^w(x,h^2D_x)-\widetilde V^w(x,h^2D_x)\right) I_l$, and let $\widetilde e(x,\xi,z)$ be the principal symbol of $\widetilde E_{-+}(z)$. Here $I_l$ denotes the unit matrix of order $l$. 
It follows from  \eqref{e1200}   that $\vert {\rm det} \,\widetilde e(x,\xi,z)\vert >c^l$.
Then for sufficiently small $h >0$,  the operator 
$\widetilde E_{-+}(z)$
is elliptic, and $\widetilde E_{-+}(z)^{-1}$ is well defined and
holomorphic for $z$ in some fixed complex neighbourhood of supp$\widetilde{f}$, (see chapter 7 of  \cite{DiSj99}). Hence, by an integration by parts, we get
$$
-{1\over\pi }\int \overline \partial_z \widetilde{f}(z)E_+
(z)\widetilde E_{-+}(z)^{-1}E_-( z)L(dz)=0.
$$
Combining this with \eqref{e5} and using the resolvent identity for
${\rm Im} \; z\neq 0$
$$
 E_{-+}(z)^{-1}=\widetilde E_{-+}(z)^{-1}+ E_{-+}(z)^{-1}
(\widetilde E_{-+}(z)- E_{-+}(z)) \widetilde E_{-+}(z)^{-1},
$$
we obtain
\begin{equation}\label{e14}
{\rm tr}\left(f(P(h))\right)=-\frac{1}{\pi}{\rm tr}\left( \int \overline{\partial}_z\widetilde{f}(z)E_+(z){E}_{-+}(z)^{-1}(\widetilde{E}_{-+}(z)-{E}_{-+}(z))\widetilde{E}_{-+}(z)^{-1}E_-(z)L(dz)\right).
\end{equation}

Since the symbol of $E_{-+}(z)-\widetilde E_{-+}(z)$ is $(\widetilde V-V)I_l$ belonging to
$ C^\infty_0(\mathbb R^2;M_l(\mathbb{C}))$, we have  $E_{-+}(z)-\widetilde E_{-+}(z) $ is a trace class  operator. It is 
then clear that we can permute integration and the operator "tr" in  the right hand side of \eqref{e14}.

Using the property of cyclic invariance  of the trace, and applying \eqref{e1404}
we get
$$
{\rm tr}\left(E_+( z) E_{-+}(z)^{-1}(\widetilde E_{-+}(z)-E_{-+}(z))\widetilde E_{-+}(z)^{-1}E_-
(z)\right)=$$
$${\rm tr}\left(E_{-+}(z)^{-1}(\widetilde E_{-+}(z)-E_{-+}(z))\widetilde E_{-+}(z)^{-1}
\partial_z E_{-+}(z)\right).
$$
Let $\chi\in C_0^\infty (\mathbb R^{2})$ be
equal to $1$ in a neighbourhood of ${\rm supp\,}( \widetilde V-V)$.
From the composition formula for two $h^2-\Psi$DOs with Weyl symbols (see Proposition \ref{p291102}), 
we see that all
the derivatives of the symbol of the operator 
$(E_{-+}(z)-\widetilde{E}_{-+}(z))(\widetilde E_{-+}(z))^{-1}\partial_z E_{-+}(z)
(1-\chi^w(x,h^2 D_x))$ 
are ${\mathcal O}(h^{2N}\langle (x,\xi )\rangle ^{-N})$
for every $N\in {\mathbb N}$. The trace class-norm
of this expression is therefore ${\mathcal O}(h^\infty )$, and consequently
\begin{equation}\label{e1400}
{\rm tr}(E_+( z)E_{-+}(z)^{-1}(\widetilde E_{-+}(z)-E_{-+}(z))\widetilde
E_{-+}(z)^{-1}E_-(z))=
\end{equation}
$$
{\rm tr}(E_{-+}(z)^{-1}(\widetilde E_{-+}(z)-E_{-+}(z))
\widetilde E_{-+}(z)^{-1}\partial_z E_{-+}(z)\chi^w(x,h^2D_x))+{\mathcal O}(h^\infty|{\rm Im}z|^{-1}).
$$
Here we recall from \eqref{e1406} that $E_{-+}(z)^{-1}={\mathcal O}(|{\rm Im}z|^{-1})$.

Inserting  \eqref{e1400}  into   \eqref{e14}, and using the fact that $\widetilde E_{-+}(z)^{-1}\partial_z E_{-+}(z)$ is holomorphic in $z$ we
obtain   \eqref{e12}. The proof of    \eqref{e1414}  is similar.

\end{proof}

Trace formulas involving effective Hamiltonian  like \eqref{e12} and \eqref{e1414}  were studied in \cite{Di98, Di01}. Applying Theorem 1.8  in \cite{Di98} to the left hand side of \eqref{e12}, we obtain
\begin{equation}\label{e1400000}
{\rm tr}(f(P(h)))
\sim\sum_{j=0}^\infty \beta_jh^{j-2},\,\, (h\searrow 0).
\end{equation}

To use Theorem 1.8 in  \cite{Di98} we make the following definition.

\begin{definition}{\rm We say that $p(x,\xi) \in S^0(\mathbb R^2;M_l(\mathbb C)),$ is micro-hyperbolic at $(x_0,\xi_0)$ in the direction $T\in \mathbb R^2$, if there are constants $C_0,C_1,C_2>0$ such that
$$(\langle dp(x,\xi),T\rangle \omega,\omega)\geq  \frac{1}{C_0} \Vert \omega\Vert^2 -C_1 \Vert p(x,\xi) \omega\Vert^2.$$
for all $(x,\xi)\in\mathbb R^2$ with $\Vert (x,\xi)-(x_0,\xi_0)\Vert \leq \frac{1}{C_2}$ and all $\omega \in \mathbb C^l$.}
\end{definition}

The assumption of Theorem \ref{t291102} implies that the principal symbol $e^0_{-+}(x,\xi,z)$ of $E_{-+}(z)$ is micro-hyperbolic at every point $(x_0,\xi_0)\in \Sigma_\mu:=\{(x,\xi)\in \mathbb R^2;\,\,
{\rm det}(e^0_{-+}(x,\xi,\mu))=0\}$.
Thus,  according to Theorem 1.8 in \cite{Di98} there exists $C>0$ large enough and $\epsilon>0$ small such that  for $f\in C_0^\infty(]\mu-\epsilon,\mu+\epsilon[;\mathbb{R})$, $\theta\in C^\infty_0(]-\frac{1}{C},\frac{1}{C}[;\mathbb R)$, we have:
\begin{equation}\label{e14000000}
{\rm tr}\left(f(P(h))\breve\theta_{h^2}(t-P(h))\right)
\sim\sum_{j=0}^\infty \gamma_j(t)h^{j-2},\,\, (h\searrow 0),
\end{equation}
with $\gamma_0(t)=c_0(t)$.

By observing that the  $h$-pseudo-differential calculus can be extended to $h<0$,  we have
$$\left\vert h^2{\rm tr}\left(f(P(h))\breve\theta_{h^2}(t-P(h))\right)-\sum_{0\leq j\leq N} \gamma_j(t) h^j\right\vert \leq C_N \vert h\vert^{N+1},\,\,\, h\in ]-h_N,h_N[\setminus\{0\}.$$
$$\left\vert  h^2{\rm tr} (f(P(h))-\sum_{0\leq j\leq N} \beta_j h^j\right\vert \leq C_N \vert h\vert^{N+1},\,\,\, h\in ]-h_N,h_N[\setminus\{0\}.$$
By the change of variable $(x,y)\rightarrow (x,-y)$, we see that $P(h)$ is unitarily equivalent to $P(-h)$. 
From this we deduce that 
 $h^2{\rm tr}(f(P(h)))$  and   $h^2{\rm tr} \left(f(P(h))\breve\theta_{h^2}(t-P(h))\right)$  are  unchanged when we replace  $h$ by $-h$. We recall that if $A$ and $B$ are unitarily equivalent trace class operators then  ${\rm tr}(A)={\rm tr}(B)$. Consequently, $ \gamma_{2j+1}=\beta_{2j+1}=0$.
This ends the proof of Theorem \ref{t3} and Theorem \ref{t291102}.  

\subsection{Proof of Corollary \ref{t5}.}
Pick $\sigma>0$  small enough. Let  $\phi_1\in C_0^\infty
(\left(a -\sigma ,a+\sigma \right) ;[0,1] )$,
$\phi_2\in C_0^\infty
(\left(a +{\sigma\over 2} ,b -{\sigma\over 2}
\right) ;[0,1] )$,
$\phi_3\in C_0^\infty
(\left(b -\sigma ,b  +\sigma \right) ;[0,1])$ satisfy $\phi_1+\phi_2+\phi_3=1$
 on $\left(
a-{\sigma \over 2},b +{\sigma \over 2}\right) $.
Let  $\gamma _0(h)\le \gamma _1(h)\le \cdots\le \gamma 
_{N}(h)$ be the eigenvalues of $H(h)$
counted with their multiplicity and lying in the interval
$\left(a -\sigma, b+\sigma \right)$. We have
\begin{align}\label{e1111}
\begin{split}
{\mathcal N}_h(a,b)&=\sum_{a 
\leq
\gamma _j(h)\leq b }(\phi_1+\phi_2+\phi_3)(\gamma _j(h))\\
 &=\sum_{a \leq \gamma _j(h)}\phi_1(\gamma
_j(h))+\sum \phi_2(\gamma _j(h))+\sum_{\gamma _j(h)\leq
b}\phi_3(\gamma _j(h))\\
&=\sum_{a \leq \gamma _j(h)}\phi_1(\gamma
_j(h))+{\rm tr}(\phi_2(H(h)))+\sum_{\gamma _j(h)\leq
b }\phi_3(\gamma _j(h)).
\end{split}
\end{align}
According to Theorem \ref{t3}, we have
\begin{equation}\label{e11110}
{\rm tr}\left(\phi_m(H(h))\right)=
\frac{1}{2\pi h^2}\sum_{j=1}^l\int_{{\mathbb R}^{2}}\phi_m((2j+1)+V(X))dX 
+{\mathcal O}(1), \quad m=1,2,3.
\end{equation}
Set 
$M(\tau,h):=
\sum\limits_{\gamma _j(h)\leq
\tau }\phi_3(\gamma _j(h))$.
Evidently, in the sense of distribution, we have
\begin{equation}\label{e111101}
{\mathcal M}(\tau):=M'(\tau,h)=
\sum_j\delta(\tau-\gamma_j(h))\phi_3(\gamma_j(h)).
\end{equation}
 In what follows, we choose   $\theta\in 
C^\infty_0(\left( -{1\over C},{1\over
C}\right) ;\lbrack 0,1 \rbrack)$, ($C>0$ large enough) such that
$ \theta(0)=1, \,\,
\breve\theta(t)\geq 0, t\in {\mathbb R},\,\,
\breve\theta(t)\geq \epsilon_0, t\in \lbrack -\delta_0,
\delta_0\rbrack \quad {\rm for \; some} \quad \delta_0>0, \epsilon_0>0.
$

\begin{corollary}\label{c061201}
There is  $C_0>0$, such that, for all  $(\lambda,h)\in {\mathbb R}\times 
\left( 0,h_0\right)$,  we have:
$$\vert M(\lambda+\delta_0h^2,h)-M(\lambda-\delta_0h^2,h)\vert\leq 
C_0.$$
\end{corollary}
\begin{proof} Since $\phi_3\ge 0$, it follows from the construction of $\theta$  that
$$
{\epsilon_0\over
h^2}\sum_{\lambda-\delta_0
h^2\leq \gamma_j(h)\leq
\lambda+\delta_0h^2}\phi_3(\gamma_j(h))\leq
\sum_{\vert
{\lambda-\gamma_j(h)
}\vert< \delta_0 h^2}
\breve\theta_{h^2}(\lambda-\gamma_j(h))\phi_3(\gamma_j(h)) \leq
$$
$$
\sum_j \breve\theta_{h^2}(\lambda-\gamma_j(h))\phi_3(\gamma_j(h))=
\breve\theta_{h^2}\star{\mathcal M}(\lambda)={\rm tr}\left(\phi_3(H(h))
\breve\theta_{h^2}(\lambda-H(h))\right). 
$$
Now Corollary \ref{c061201}  follows from \eqref{e290561}.
 \end{proof}
According to Corollary \ref{c061201}, we have
\begin{equation}\label{e11110101}
\int\left\langle{\tau-\lambda
\over h^2}\right\rangle^{-2}
{\mathcal M}(\tau)d\tau=\sum_{k\in {\mathbb Z}}\int_{\{
\delta_0k\leq {\tau-\lambda\over h^2}\leq \delta_0(k+1)\}}\left\langle{\tau-\lambda
\over h^2}\right\rangle^{-2}
{\mathcal M}(\tau)d\tau\leq  C_0
\left(\sum_{k\in
{\mathbb Z}}\langle {\delta_0k}\rangle^{-2}\right).
\end{equation}
On the other hand, since  
$\breve\theta \in {\mathcal S}({\mathbb R})$ and  $\theta(0)=1$, there exists
 $C_1>0$ such that:
$$\left\vert \int_{-\infty}^\lambda
\breve\theta_{h^2}(\tau-y)dy-
1_{\left(
-\infty,\lambda\right)}(\tau)\right\vert=\left\vert \int_{{\tau-\lambda\over 
h^2}}^{+\infty}\breve\theta(y)dy-
1_{\left(-\infty,\lambda\right)}(\tau)\right\vert \leq C_1
\left\langle\frac{\tau-\lambda}{h^2}\right\rangle^{-2},$$
uniformly in $\tau\in {\mathbb R }$ and  $h\in \left( 0,h_0\right)$.
 Consequently, 
\begin{equation}\label{e1111010}
\left\vert \int_{-\infty}^\lambda\breve\theta_{h^2}\star
{\mathcal M}(\tau)d\tau-
\int_{-\infty}^\lambda{\mathcal M}
(\tau)d\tau\right\vert \le C_1 \int
\left\langle{\tau-\lambda\over
h^2}\right\rangle^{-2}
{\mathcal M}(\tau)d\tau.
\end{equation}
Putting together
(\ref{e111101}), (\ref{e11110101}) and (\ref{e1111010}), we  get 
\begin{equation}\label{e11110000}
\int_{-\infty}^
\lambda\breve\theta_{h^2}\star
{\mathcal M}(\tau)d\tau=M(\lambda,h)+{\mathcal
O}(1).
\end{equation}
Note that  $\breve\theta_{h^2}\star{\mathcal M}(\tau)={\rm tr}\left(\phi_3(H(h))
\breve\theta_{h^2}(\tau-H(h))\right)$.
As a consequence of   \eqref{e290561},   \eqref{e061201}   and \eqref{e11110000} we obtain 
\begin{equation}\label{e1111001}
M(\lambda,h)=
h^{-2}m(\lambda)+{\mathcal O}(1),
\end{equation}
where
\begin{equation}\label{e11110011}
m(\lambda)=\int_{-\infty}^\lambda
c_0(\tau)d\tau={1\over 2\pi}\sum_{j=1}^l\int_{\{X\in 
{\mathbb R}^{2}| (2j+1)+V(X)\leq\lambda\}} \phi_3((2j+1)+V(X)) dX.
\end{equation}
Here we have used the fact that  if  $E$ is not a critical value of 
$V(X)$, then
$${\partial\over \partial E}\left(\int_{\{X\in {\mathbb R}^{2};\;V(X)\le E\}} 
\phi(V(X))dX \right)=\phi(E)\int_{S_E} {dS_E\over \vert \nabla_{X} V\vert},$$
where $S_E=V^{-1}(E)$ (see \cite[Lemma V-9]{Ro87}).\\
Applying  \eqref{e29055},  (\ref{e1111001}) and (\ref{e11110011}) to  $\phi_1$ and writing:
$$\sum_{a \leq \gamma
_j(h)}\phi_1(\gamma_j(h))=\sum_j\phi_1(\gamma_j(h))-\sum_{
\gamma _j(h)<a}\phi_1(\gamma_j(h)),$$
we get 
\begin{equation}\label{e111102}
\sum_{a \leq
\gamma _j(h)}\phi_1(\mu
_j(h))=h^{-2}m_1(a)+{\mathcal O}(1),
\end{equation}
with
\begin{equation}\label{e111103}
m_1(a)= {1\over 2\pi}\sum_{j=1}^l\int_{\{X\in 
{\mathbb R}^{2}| (2j+1)+V(X)\geq a\}} \phi_1((2j+1)+V(X)) dX.
\end{equation}
Now Corollary \ref{t5} results  from  \eqref{e11110}, \eqref{e111101}, \eqref{e1111001}, \eqref{e11110011}, \eqref{e111102}  and \eqref{e111103}.

\section{Proof of  Theorem \ref{t1610} and Theorem \ref{t7}}\label{s5}
As we have noticed in  the outline of the proofs, we will construct a potential
$$\varphi(X;h)=\varphi_0(X)+\varphi_2(X) h^2+\cdots +\varphi_{2j}(X)h^{2j}+\cdots,$$
  such that for all $f\in C^\infty_0((a,b);\mathbb R)$ and $\theta \in C^\infty_0(\mathbb R;\mathbb R)$,  we have
\begin{equation}\label{eL1}
{\rm tr}(f(H_\lambda))={\rm tr}(f(Q))+{\mathcal O}(h^\infty),
\end{equation}
\begin{equation}\label{eL2}
{\rm tr}\left(f(H_\lambda) \breve \theta_{\lambda^{-\frac{2}{\delta}}}(t-H_\lambda)\right)={\rm tr}\left(f(Q)\breve\theta_{h^2}(t-Q)\right)+{\mathcal O}(h^\infty),
\end{equation}
where $Q:=H_0+\varphi(hX;h)$ and $h=\lambda^{-\frac{1}{\delta}}$.   By observing that 
  Theorem \ref{t3}, Theorem \ref{t291102} and Corollary \ref{t5} remain true when we replace $H(h)=H_0+V(hX)$ by
$Q$,  Theorem \ref{t1610}, Theorem \ref{t291101} and Corollary \ref{t7}  follow from (\ref{eL1}) and (\ref{eL2}).  
The remainder of this paper is devoted to the proof of  \eqref{eL1} and \eqref{eL2}.
\subsection{Construction of  reference operator  $Q$}
Set $h=\lambda^{-\frac{1}{\delta}}$. For $M>0$, 
 put
\begin{equation}
 \Omega_M(h)=\{X\in\mathbb{R}^{2};\;h^{-\delta}V(X)>M\}.
\end{equation}
Since $\omega_0>0$ and continuous on the unit circle, 
there exist two positive constants $C_1$ and $C_2$ such that $C_1<(\min\limits_{\mathbb{S}^1} \omega_0)^{1/\delta}\leq
 (\max\limits_{\mathbb{S}^1} \omega_0)^{1/\delta}<C_2$. 

According to the hypothesis \eqref{e11}, there exists $h_0>0$ such that 
$$ B(0,C_1M^{-1/\delta}h^{-1})\subset\Omega_M(h)\subset B(0,C_2M^{-1/\delta}h^{-1}),\;\mbox{ for all }0<h\leq h_0.$$ Here 
$B(0,r)$ denotes the ball of center $0$ and radius $r$.

 Let  $\chi\in C_0^\infty(B(0,C_1M^{-1/\delta});[0,1])$ satisfying  $\chi=1$ near zero. Set 
\begin{itemize}
\item $\varphi(X;h):=(1-\chi(X))h^{-\delta}V(\frac{X}{h})+M\chi(X),$

\item  $W_h(X):=h^{-\delta}V(X)-\varphi(hX;h)=\chi(hX)(h^{-\delta}V(X)-M)$.
\end{itemize}

By the construction of $\varphi(\cdot ;h)$ and $W_h$, we have
\begin{equation}\label{e111001}
|\partial_X^\alpha\varphi(X;h)|\leq C_\alpha, \mbox{  uniformly for  }h\in(0,h_0],
\end{equation}
\begin{equation}\label{e111002}
\varphi(hX;h)>\frac{M}{2}\mbox{  for  }X\in\Omega_\frac{M}{2}(h),
\end{equation}
\begin{equation}
{\rm supp}W_h\subset B(0,C_1M^{-1/\delta}h^{-1})\subset\Omega_M(h).
\end{equation}
On the other hand, it follows from \eqref{e11} that for all $N\in\mathbb{N}$, there exist
$\varphi_0,...,\varphi_{2N},K_{2N+2}(\cdot ;h)\in C^\infty(\mathbb{R}^2;\mathbb{R})$, uniformly
 bounded with respect to $h\in(0,h_0]$ together with their derivatives such that:
\begin{equation}
\varphi(X;h)= \sum\limits_{j= 0}^N\varphi_{2j}(X)h^{2j}+h^{2N+2}K_{2N+2}(X;h)
\end{equation}
 with 
$$\varphi_0(X)=(1-\chi(X))\omega_0\left(
\frac{X}{|X|}\right)|X|^{-\delta}+M\chi(X).$$
In fact,  if $X\in {\rm supp}\chi$ then  $\omega_0\left(
\frac{X}{|X|}\right)|X|^{-\delta}>C_1^\delta|X|^{-\delta}>M$, which implies that
 $\varphi_0(X)\geq (1-\chi(X))M+M\chi(X)=M$  for all $X\in {\rm supp}\chi$. Consequently, we have
\begin{lemma}\label{l291101}
If $\varphi_0(X)<M$ then $\varphi_0(X)=\omega_0\left(
\frac{X}{|X|}\right)|X|^{-\delta}$.
\end{lemma}
Let  $\psi\in C^\infty(\mathbb{R};[\frac{M}{3},+\infty))$  satisfying $\psi(t)=t$ for all $t\geq \frac{M}{2}$. We define
$$F_1(X;h)
:=\psi(\varphi(hX;h))\mbox{  and  }
 F_2(X;h):=\psi(h^{-\delta}V(X)).$$
Let $\mathcal{U}$ be a small complex neighborhood of $[a,b]$. From now on, we choose 
 $M>a+b$ large enough such that
\begin{equation}
F_j(X;h)-{\rm Re}z\geq \frac{M}{4},\;j=1,2,
\end{equation}
uniformly for $z\in\mathcal{U}$.
This choice of M implies that:
\begin{itemize}
\item If $2j+1+\varphi(X;h)\in [a,b]$  then $\varphi_0(X)<M$ for all $h\in(0,h_0]$,
\item The function defined by $z\mapsto(z-H_{F_j})^{-1}$ is holomorphic from $\mathcal{U}$
 to $\mathcal{L}(L^2(\mathbb{R}^2))$, where $H_{F_j}:=H_0+F_j(X;h),\;j=1,2.$
\end{itemize}
Moreover, it follows from \eqref{e111001} that $\partial_X^\alpha F_j(X;h)=\mathcal{O}_\alpha(h^{-\delta}).$

Finally, \eqref{e111002} shows that
\begin{equation}\label{e221001}
\begin{split}
 &{\rm dist}\left({\rm supp}W_h,{\rm supp}[\varphi(h\cdot ;h)-F_1(\cdot ;h)]\right)\geq\frac{a_1(M)}{h},\\
& {\rm dist}\left({\rm supp}W_h,{\rm supp}[h^{-\delta}V(\cdot)-F_2(\cdot ;h)]\right)\geq\frac{a_2(M)}{h},
\end{split}
\end{equation}
with $a_1(M), a_2(M)>0$ independent of $h$.
\begin{lemma}\label{l22051}
Let $\widetilde{\chi}\in C_0^\infty(\mathbb{R}^2)$. For $z\in\mathcal{U},$ the operators
$\widetilde{\chi}(hX)(z-H_{F_j})^{-1}$, $j=1,2,$ belong to the class of 
Hilbert-Schmidt operators. Moreover
\begin{equation}\label{e22051}
\|\widetilde{\chi}(hX)(z-H_{F_j})^{-1}\|_{\rm HS}=\mathcal{O}(h^{-3-\delta}).
\end{equation}
Here we denote by $\|.\|_{\rm HS}$ the Hilbert-Schmidt norm of operators.
\end{lemma}
\begin{proof}
We prove \eqref{e22051} for $j=1$. The case $j=2$ is treated
in the same way.

Using the resolvent equation, one has
\begin{equation}\label{e22052}
(z-H_{F_1})^{-1}=\left(z-{M\over 6}-H_0\right)^{-1}+
\left(z-\frac{M}{6}-H_0\right)^{-1}\left(F_1(X;h)-\frac{M}{6}\right)(z-H_{F_1})^{-1}.
\end{equation}
 On the other hand, the 
operator $\left(z-\frac{M}{6}-H_0\right)^{-1}$ was shown to be
an integral operator with integral kernel $K_0(X,Y,z)$
 satisfying $|K_0(X,Y,z)|\leq C e^{-\frac{1}{8}|X-Y|^2}$ uniformly for $z\in \mathcal{U}$ (see \cite[Formula 2.17]{CN98}). Let $K_1(X,Y,z)$ be the integral kernel of
$\widetilde{\chi}(hX)(z-\frac{M}{6}-H_0)^{-1}$. Then  $K_1(X,Y,z)=\widetilde{\chi}(hX)K_0(X,Y,z)$.

Let $\langle X\rangle=(1+|X|^2)^{\frac{1}{2}}$, $X\in\mathbb{R}^2$. Since $\widetilde{\chi}\in C_0^\infty(\mathbb{R}^2)$,
 one has $\widetilde{\chi}(hX)h^3\langle X\rangle^3$ is uniformly bounded for $h>0$. 
 Combining this with the fact that 
$\frac{1}{\langle X\rangle^3}e^{-\frac{1}{8}|X-Y|^2}\in 
L^2(\mathbb{R}^4)$, we obtain
\begin{equation}\label{e22053}
\|K_1(X,Y,z)\|_{L^2(\mathbb{R}^4)}=\left\|\widetilde{\chi}(hX)h^3\langle X\rangle^3\frac{1}{h^3\langle X\rangle^3}K_0(X,Y,z)\right\|_{L^2(\mathbb{R}^4)}=\mathcal{O}(h^{-3}).
\end{equation}
It shows that $\widetilde{\chi}(hX)(z-\frac{M}{6}-H_0)^{-1}$ is  a Hilbert-Schmidt operator and 
\begin{equation}\label{e22054}
\left\|\widetilde{\chi}(hX)\left(z-\frac{M}{6}-H_0\right)^{-1}\right\|_{\rm HS}=
\left\|K_1(X,Y,z)\right\|_{L^2(\mathbb{R}^4)}=\mathcal{O}(h^{-3}).
\end{equation}
Consequently,  \eqref{e22052} and \eqref{e22054} imply that
\begin{align*}
&\|\widetilde{\chi}(hX)(z-H_{F_1})^{-1}\|_{\rm HS}\leq 
\left\|\widetilde{\chi}(hX)\left(z-\frac{M}{6}-H_0\right)^{-1}\right\|_{\rm HS}\\
&+\left\|\widetilde{\chi}(hX)\left(z-\frac{M}{6}-H_0\right)^{-1}\right\|_{\rm HS}\left\|F_1(X;h)-\frac{M}{6}\right\|_{L^\infty(\mathbb{R}^2)}\|(z-H_{F_1})^{-1}\|\\
&=\mathcal{O}(h^{-3-\delta}),
\end{align*}
where we have used $F_1(X;h)=\mathcal{O}(h^{-\delta})$.
\end{proof}
\begin{lemma}\label{l22052}
For $z\in\mathcal{U}$, the operator 
$$W_h(X)(H_{F_1}-z)^{-1}(\varphi(hX;h)-F_1(X;h))$$
belongs to the class of Hilbert-Schmidt operators. Moreover,
\begin{equation}\label{e25055}
\|W_h(X)(H_{F_1}-z)^{-1}(\varphi(hX;h)-F_1(X;h))\|_{\rm HS}=\mathcal{O}(h^\infty).
\end{equation}
\end{lemma}
\begin{proof}
Let $H^0_{F_1}:=-\Delta+F_1(X;h)$. We denote by $G(X,Y;z)$ (resp. $G_0(X,Y;{\rm Re }z)$)  the Green function of $(H_{F_1}-z)^{-1}$ (resp. $(H^0_{F_1}-{\rm Re }z)^{-1}$).

From the functional calculus, one has
\begin{align}\label{e25056}
\begin{split}
(H_{F_1}-z)^{-1}&=\int_0^\infty e^{tz}e^{-tH_{F_1}}dt,\\
(H^0_{F_1}-{\rm Re }z)^{-1}&=\int_0^\infty e^{t{\rm Re }z}e^{-tH^0_{F_1}}dt.
\end{split}
\end{align}
For $t\geq 0$, the Kato
 inequality (see \cite[Formula 1.8]{CFKS87}) implies that
\begin{equation}\label{e25057}
|e^{-tH_{F_1}}u|\leq e^{-tH^0_{F_1}}|u| \;({\text{pointwise}}),\;u\in L^2(\mathbb{R}^2).
\end{equation}
Then \eqref{e25056} and \eqref{e25057} yield
\begin{equation}\label{e25058}
|(H_{F_1}-z)^{-1}u|\leq (H^0_{F_1}-{\rm Re}z)^{-1}|u|\;({\text{pointwise}}),\;u\in L^2(\mathbb{R}^2).
\end{equation}
Consequently, applying  \cite[Theorem 10]{BGP07} we have $|G(X,Y;z)|\leq G_0(X,Y;{\rm Re}z)$ for a.e. $X,Y\in\mathbb{R}^2$.
From this, one obtains 
\begin{equation}\label{e111003}
|W_h(X)G(X,Y;z)(\varphi(hY;h)-F_1(Y;h))|\leq |W_h(X)G_0(X,Y;{\rm Re}z)(\varphi(hY;h)-F_1(Y;h))|
\end{equation}
for a.e. $X,Y\in\mathbb{R}^2$. 

On the other hand, using \eqref{e221001} M. Dimassi proved that
 (see  \cite[Proposition 3.3]{Di94})
\begin{equation}\label{e111004}
\|W_h(X)G_0(X,Y;{\rm Re}z)(\varphi(hY;h)-F_1(Y;h))\|_{L^4(\mathbb{R}^4)}=\mathcal{O}(h^\infty).
\end{equation}
 Thus, \eqref{e111003} and \eqref{e111004} give
\begin{equation}\label{e25059}
\|W_h(X)G(X,Y;z)(\varphi(hY;h)-F_1(Y;h))\|_{L^4(\mathbb{R}^4)}=\mathcal{O}(h^\infty).
\end{equation}
The estimate \eqref{e25059} shows that the operator $W_h(X)(H_{F_1}-z)^{-1}(\varphi(hX;h)-F_1(X;h))$ is Hilbert-Schmidt and 
\begin{equation}\label{e250510}
\|W_h(X)(H_{F_1}-z)^{-1}(\varphi(hX;h)-F_1(X;h))\|_{\rm HS}=\mathcal{O}(h^\infty).
\end{equation}
\end{proof}
By using the same arguments as in  Lemma \ref{l22052}, we also obtain
\begin{lemma}
For $z\in\mathcal{U}$, the operator $$W_h(X)(H_{F_2}-z)^{-1}(h^{-\delta}V(X)-F_2(X;h))$$
belongs to the class of Hilbert-Schmidt operators and
$$\|W_h(X)(H_{F_2}-z)^{-1}(h^{-\delta}V(X)-F_2(X;h))\|_{\rm HS}=\mathcal{O}(h^\infty).$$ 
\end{lemma}
Let $Q:=H_0+\varphi(hX;h)$. For $z\in\mathcal{U}$, ${\rm Im}z\not=0$, put
\begin{equation}\label{e121107}
G(z)=(z-H_\lambda)^{-1}-(z-Q)^{-1}-(z-H_{F_2})^{-1}W_h(z-H_{F_1})^{-1}.
\end{equation}
\begin{proposition}\label{p291101}
The operator $G(z)$ is of trace class and satisfies the following estimate:
\begin{equation}\label{e111005}
\|G(z)\|_{\rm tr}=\mathcal{O}(h^\infty|{\rm Im}z|^{-2}),
\end{equation}
uniformly for $z\in \mathcal{U}$ with ${\rm Im}z\not=0$.
\end{proposition}
\begin{proof}
It follows from the resolvent equation  that
\begin{equation}\label{e221002}
(z-H_\lambda)^{-1}-(z-Q)^{-1}=(z-H_\lambda)^{-1}W_h(z-Q)^{-1}.
\end{equation}
On the other hand, one has
\begin{align}\label{e221003}
\begin{split}
 (z-H_\lambda)^{-1}&=(z-H_{F_2})^{-1}\\
&+(z-H_\lambda)^{-1}(h^{-\delta}V(X)-F_2(X;h))(z-H_{F_2})^{-1}
\end{split}
\end{align}
and
\begin{align}\label{e221004}
\begin{split}
(z-Q)^{-1}&=(z-H_{F_1})^{-1}\\
&+(z-H_{F_1})^{-1}(\varphi(hX;h)-F_1(X;h)) (z-Q)^{-1}.
\end{split}
\end{align}
 Substituting  \eqref{e221003} and \eqref{e221004} into the right hand side of
  \eqref{e221002}, one gets
\begin{align*}
&G(z)=(z-H_{F_2})^{-1}W_h(z-H_{F_1})^{-1}(\varphi(hX;h)-F_1(X;h))
 (z-Q)^{-1}\\ 
&+(z-H_\lambda)^{-1}(h^{-\delta}V(X)-F_2(X;h))(z-H_{F_2})^{-1}W_h(z-H_{F_1})^{-1}\\
&+(z-H_\lambda)^{-1}(h^{-\delta}V(X)-F_2(X;h))(z-H_{F_2})^{-1}W_h(z-H_{F_1})^{-1}\times\\
&\times (\varphi(hX;h)-F_1(X;h)) (z-Q)^{-1}=:A(z)+B(z)+C(z).
\end{align*}
Next we choose  
$\widetilde{\chi}\in C_0^\infty(\mathbb{R}^2)$ such that $\widetilde{\chi}(hX)W_h(X)=W_h(X)$. 
It follows from Lemma \ref{l22051} and Lemma \ref{l22052} that
\begin{align*}
&\nonumber\|A(z)\|_{\rm tr}\leq \|(z-H_{F_2})^{-1}\widetilde{\chi}(hX)\|_{\rm HS}\|W_h(z-H_{F_1})^{-1}(\varphi(hX,h)-F_1(X;h)) \|_{\rm HS}\|(z-Q)^{-1}\|\\
&=\mathcal{O}(h^\infty{|{\rm Im}z|}^{-1}).
\end{align*}
Here we have used the fact that $\|(z-Q)^{-1}\|=\mathcal{O}({|{\rm Im}z|}^{-1})$.
Similarly, we also obtain $\|B(z)\|_{\rm tr}=\mathcal{O}(h^\infty{|{\rm Im}z|}^{-1})$ and 
  $\|C(z)\|_{\rm tr}=\mathcal{O}(h^\infty{|{\rm Im}z|}^{-2})$. Thus, 
$$ \|G(z)\|_{\rm tr} \leq \|A(z)\|_{\rm tr} +\|B(z)\|_{\rm tr}+\|C(z)\|_{\rm tr}=\mathcal{O}(h^\infty{|{\rm Im}z|}^{-2}).$$
\end{proof}
\subsection{Proof of  \eqref{eL1} and Theorem \ref{t1610}}
Let $f\in C_0^\infty((a,b);\mathbb{R})$ and let $\widetilde f\in C_0^\infty(\mathcal{U})$ be
  an almost analytic extension of $f$. From the  Helffer- Sj\"{o}trand formula and   \eqref{e121107},  we get
\begin{align}\label{ebss2}
\begin{split}
&f(H_\lambda)-f(Q)\\
&=-\frac{1}{\pi}\int\overline{\partial}_z\widetilde f(z)
[(z-H_\lambda)^{-1}-(z-Q)^{-1}]L(dz)\\
&=-\frac{1}{\pi}\int\overline{\partial}_z\widetilde f(z)\left[(z-H_{F_2})^{-1}W_h(z-H_{F_1})^{-1}+G(z)\right]L(dz).
\end{split}
\end{align}
Notice that $(z-H_{F_2})^{-1}W_h(z-H_{F_1})^{-1}$ is holomorphic
 in $z\in\mathcal{U}$, then
\begin{equation}\label{e121104}
-\frac{1}{\pi}\int\overline{\partial}_z\widetilde f(z)(z-H_{F_2})^{-1}W_h(z-H_{F_1})^{-1}L(dz)=0.
\end{equation}
Thus, \eqref{ebss2} and \eqref{e121104} follow that
\begin{equation}\label{e121105}
f(H_\lambda)-f(Q)=-\frac{1}{\pi}\int\overline{\partial}_z\widetilde f(z)G(z)L(dz),
\end{equation}
which together with \eqref{e111005}  yields \eqref{eL1}. 

Applying Theorem \ref{t3} to the operator $Q$ and using \eqref{eL1} we obtain \eqref{e121101} with
$$b_0(f)=\sum_{j=0}^q\frac{1}{2\pi}\iint f(2j+1
+\varphi_0(X))dX$$
According to Lemma   \ref{l291101}, one has
$$2j+1
+\varphi_0(X)\in [a,b] \Longleftrightarrow \varphi_0(X)=\omega_0\left(
\frac{X}{|X|}\right)|X|^{-\delta}, j=0,...,q.$$
Thus, after a change of variable in the integral we get
$$b_0(f)=\frac 1{2\pi\delta}\int_0^{2\pi} (\omega_0(\cos \theta,\sin \theta))^{\frac 2{\delta}} d\theta \sum_{j=0}^q\int  f(u)(u-(2j+1))^{-1-\frac 2{\delta}} du.$$
We recall that  ${\rm supp} f\subset ]a,b[$, with $2q+1<a<b<2q+3$.  This ends the proof of Theorem \ref{t1610}.
\subsection{Proof of  \eqref{eL2} and Theorem \ref{t291101}}
The proof of \eqref{eL2}  is a slight modification of  \eqref{eL1}. For that, 
let $\phi\in C_0^\infty((-2,2);[0,1])$ such that $\phi=1$ 
 on $[-1,1]$. Put $\phi_h(z)=\phi(\frac{{\rm Im}z}{h^2})$, then $\widetilde f(z)\phi_h(z)$ is also an 
 almost analytic extension of $f$. Applying again the Helffer-Sj\"{o}strand formula, we get
\begin{align}\label{e121108}
\begin{split}
&f(H_\lambda)\breve\theta_{\lambda^{-\frac{2}{\delta}}}(t-H_\lambda)-f(Q)\breve\theta_{h^2}(t-Q)\\
&=-\frac{1}{\pi}\int\overline{\partial}_z(\widetilde f\phi_h)(z)
\breve\theta_{h^2}(t-z)[(z-H_\lambda)^{-1}-(z-Q))^{-1}]L(dz)\\
&=-\frac{1}{\pi}\int\overline{\partial}_z(\widetilde f\phi_h)(z)
\breve\theta_{h^2}(t-z)\left[(z-H_{F_2})^{-1}W_h(z-H_{F_2})^{-1}+G(z)\right]L(dz)\\
&=-\frac{1}{\pi}\int\overline{\partial}_z(\widetilde f\phi_h)(z)
\breve\theta_{h^2}(t-z)G(z)L(dz),
\end{split}
\end{align}
where in the last equality we have used the fact that  $(z-H_{F_2})^{-1}W_h(z-H_{F_1})^{-1}$ is holomorphic in $z\in\mathcal{U}$.

According to the  Paley-Wiener theorem (see e.g. \cite[Theorem IX.11]{ReSi2}) the function   $\breve\theta_{h^2}(t-z)$  is  analytic 
  with respect to $z$ and satisfies the following estimate
\begin{equation}\label{e121111}
\breve\theta_{h^2}(t-z)=\mathcal{O}\left(\frac{1}{h^2}\exp\left(\frac{|{\rm Im}z|}{Ch^2}\right)\right).
\end{equation}
Combining this with the fact that $\overline{\partial}_z(\widetilde f\phi_h)(z)={\mathcal O}(\vert{\rm Im} z\vert^\infty)\phi_h(z)+\mathcal{O}\left(\frac{1}{h^2}\right)1_{[h^2,2h^2]}(\vert{\rm Im} z\vert)$, and using Proposition \ref{p291101} we get
$$\| \overline{\partial}_z(\widetilde f\phi_h)(z)
\breve\theta_{h^2}(t-z)G(z)\|_{\rm tr}={\mathcal O}(h^\infty).$$
This together with \eqref{e121108} ends the proof of  \eqref{eL2}.

By observing that $X.\nabla_X\left(\omega_0\left(\frac{X}{|X|}\right)\right)=0$, we have
\begin{equation}
 X.\nabla_X\left(\omega_0\left(\frac{X}{|X|}\right)|X|^{-\delta}\right)=-\delta\omega_0\left(\frac{X}{|X|}\right)|X|^{-\delta}.
\end{equation}
 Then, since $\omega_0>0$, we  obtain $\nabla_X(\omega_0(\frac{X}{|X|})|X|^{-\delta})\not=0$ for $X\in\mathbb{R}^2\setminus\{0\}$.
It implies that the functions $2j+1+\varphi_0(X)$, $j=1,...,q$, do not have any critical values in the interval 
 $[a,b]$. Consequently,  Theorem \ref{t291101} follows from \eqref{eL2} and Theorem \ref{t291102}. 
\bibliographystyle{plain}

\end{document}